\documentclass[10pt,conference]{IEEEtran}%
\usepackage{subfigure}
\usepackage[utf8]{inputenc}
\usepackage[T1]{fontenc}
\usepackage{amsfonts}
\usepackage{amsmath}
\usepackage{amssymb}
\usepackage{graphicx}
\usepackage{graphicx}
\usepackage{citesort}%
\newtheorem{theorem}{Theorem}
\newtheorem{lemma}{Lemma}

\newtheorem{definition}{Definition}
\newtheorem{remark}{Remark}
\newtheorem{example}{Example}

\begin{document}

\title{Ergodic Layered Erasure One-Sided Interference Channels}
\pubid{~}
\specialpapernotice{~}%

%TCIMACRO{\TeXButton{Author Information}{\author{\authorblockN
%{Vaneet Aggarwal, Lalitha Sankar, A. Robert Calderbank, and H. Vincent Poor}
%\authorblockA
%{Department of Electrical Engineering, Princeton University, Princeton, NJ 08544.\\
%Email: {vaggarwa,lalitha,calderbk,poor}@princeton.edu\\}
%}}}%
%BeginExpansion
\author{\authorblockN
{Vaneet Aggarwal, Lalitha Sankar, A. Robert Calderbank, and H. Vincent Poor}
\authorblockA
{Department of Electrical Engineering, Princeton University, Princeton, NJ 08544.\\
Email: \{vaggarwa, lalitha, calderbk, poor\}@princeton.edu\\}
}%
%EndExpansion
%

%TCIMACRO{\TeXButton{Make Title}{\maketitle}}%
%BeginExpansion
\maketitle
%EndExpansion
%

%TCIMACRO{\TeXButton{Begin abstract}{\begin{abstract}}}%
%BeginExpansion
\begin{abstract}%
%EndExpansion

\footnotetext{The work of V. Aggarwal and A. R. Calderbank was
supported in part by NSF under grant 0701226, by ONR under grant
N00173-06-1-G006, and by AFOSR under grant FA9550-05-1-0443. The
work of L. Sankar and H. V. Poor was supported in
 part by the National Science Foundation under grant CNS-06-25637.}The sum capacity of a class of layered erasure one-sided
interference channels is developed under the assumption of no
channel state information at the transmitters. Outer bounds are
presented for this model and are shown to be tight for the
following sub-classes: i) weak, ii) strong (mix of strong but not
very strong (SnVS) and very strong (VS)), iii) \textit{ergodic
very strong} (mix of strong and weak), and (iv) a sub-class of
mixed interference (mix of SnVS and weak). Each sub-class is
uniquely defined
by the fading statistics.%

%TCIMACRO{\TeXButton{End abstract}{\end{abstract}}}%
%BeginExpansion
\end{abstract}%
%EndExpansion

\section{Introduction}

The capacity region of interference channels (IFCs), comprised of two or more
interfering links (transmitter-receiver pairs), remains an open problem. The
sum capacity of a non-fading two-user IFC is known only when the interference
is either stronger or much weaker at the unintended than at the intended
receiver (see, for e.g.,
\cite{cap_theorems:Sato_IC,cap_theorems:Carleial_VSIFC,cap_theorems:ShangKramerChen,cap_theorems:AR_VVV,cap_theorems:MotaKhan}%
, and the references therein). Recently, the sum capacity and optimal power
policies for two-user ergodic fading IFCs are studied in
 \cite{cap_theorems:SEP} and \cite{cap_theorems:Tuninetti} under the assumption that the
instantaneous fading channel state information (CSI) is known at all nodes. A
sum capacity analysis for $K$-user ergodic fading channels using ergodic
interference alignment is developed in
\cite{cap_theorems:Nazer01} and \cite{cap_theorems:Jafar_ErgIFC}. In general, however,
the instantaneous CSI is not available at the transmitters and often involves
feedback from the receivers. Thus, it is useful to study the case in which
only receivers have perfect CSI and the transmitters are strictly restricted
to knowledge only of the channel statistics.

The sum capacity of multi-terminal networks without transmit CSI remains a
largely open problem with the capacity known only for ergodic fading Gaussian
multiaccess channels (MACs) without transmit CSI. For this class of channels,
it is optimal for each user to transmit at its maximum average power in each
use of the channel (see for e.g.,
\cite{cap_theorems:TH01} or \cite{cap_theorems:Shamai_Wyner01}). The receiver, with
perfect knowledge of the instantaneous CSI, decodes the messages from all
transmitters jointly over all fading realizations.

Recently, the sum capacity of ergodic fading two-receiver broadcast channels
(BCs) without transmit CSI has been studied in \cite{cap_theorems:TseYatesLi_01}.
The authors first develop the sum capacity achieving scheme for an
\textit{ergodic layered erasure BC} where the channel from the source to each
receiver is modeled as a time-varying version of the binary expansion
deterministic channel introduced in \cite{cap_theorems:AvestimehrDS_01}. In
this model, the transmitted signal is viewed as a vector (layers) of bits from
the most to the least significant bits. Fading is modeled as an erasure of a
random number of least significant bits and the instantaneous erasure levels,
or equivalently the number of received layers (or levels), are assumed to be
known at the receivers. For a layered erasure fading\ BC, the authors in
\cite{cap_theorems:TseYatesLi_01} show that a strategy of signaling
independently on each layer to one receiver or the other based only on the
fading statistics achieves the sum capacity. Furthermore, the authors also
demonstrate the optimality of their achievable scheme to within 1.44 bits/s/Hz
of the capacity region for a class of high-SNR channel fading distributions.

In this paper, we introduce an ergodic fading layered erasure one-sided
(two-user) IFC in which, in each channel use, one of the receivers receives a
random number of layers from its intended transmitter while the other receiver
receives a random number of layers from both transmitters. One can view this
channel as a time-varying one-sided version of a two-user binary expansion
deterministic IFC introduced and studied in \cite{cap_theorems:Bresler_Tse}.
The model in \cite{cap_theorems:Bresler_Tse} is a subset of the class of
deterministic IFCs whose capacity region is developed in
\cite{cap_theorems:CostaElGamal_IC2}. More recently, in
\cite{cap_theorems:Aggarwal_Liu_Sabharwal_IC}, the sum capacity of a class of
one-sided two-user and three-user IFCs in which each transmitter has limited
information about its connectivity to the receivers is developed. For the
ergodic layered erasure one-sided IFC considered here, we develop outer bounds
and identify fading regimes for which the strategies of either decoding or
ignoring interference at the interfered receiver is tight. We classify the
capacity achieving regimes based on the fading statistics of the direct and
interfering links as follows: i) weak, ii) strong (mix of strong but not very strong (SnVS) and very strong (VS)), iii)
\textit{ergodic very strong} (mix of SnVS, VS, and weak), and (iv) a sub-class
of mixed interference (mix of SnVS and weak).

The paper is organized as follows. In Section \ref{Sec_2} we
introduce the channel model. In Section \ref{Sec_3}, we develop
the capacity region of a layered erasure multiple-access channel.
In Section \ref{Sec_4}, we develop outer bounds for the layered
erasure IFC and identify the regimes where these bounds are tight
using in part the results developed in Section \ref{Sec_3}. We
conclude in Section \ref{Sec_5}.

\section{\label{Sec_2}Channel Model and Preliminaries\newline}
\vspace{-.15in}A two-user IFC consists of two point-to-point
transmitter-receiver links where the receiver of each link also
receives an interfering signal from the unintended transmitter. In
a deterministic IFC, the input at each transmitter is a vector of
$q$ bits. We write $X_{k}^{q}=\left[ X_{k,1}\text{ }X_{k,2}\text{
}\ldots\text{ }X_{k,q}\right]  ^{T}$, $k=1,2,$ to denote the input
at the $k^{th}$ transmitter such that $X_{k,1}$ and $X_{k,q}$
are the most and the least significant bits, respectively.
Throughout the sequel, we refer to the bits as levels or layers,
and write the input at level $n$ for transmitter $k$ as $X_{k,n}$,
for all $n=1,2,\ldots,q$. The received signal of user $k$, is
denoted by the $q$-length vector $Y_{k}^{q}=\left[  Y_{k,1}\text{
}Y_{k,2}\text{ }\ldots\text{ }Y_{k,q}\right] ^{T}$.

Associated with each transmitter $k$ and receiver $j$ is a non-negative
integer $n_{jk}$ that defines the number of bit levels of $\mathbf{X}_{k}$
observed at receiver $j$. The maximum level supported by any link is $q$.
Specifically, an $n_{jk}$ link erases $q-n_{jk}$ least significant bits of
$X_{k}^{q}$ such that only $n_{jk}$ most significant bits of $X_{k}^{q}$ are
received as the $n_{jk}$ least significant bits of $Y_{k}^{q}$. The missing
entries $X_{k,n_{jk}+1},\ldots,X_{k,q}$ have been erased by the fading
channel. Thus, we have \cite{cap_theorems:AvestimehrDS_01}%
\begin{align}
Y_{k}^{q}  &  =\left[  0\text{ }0\text{ \ldots\ }0\text{ }X_{k,1}\text{
}X_{k,2}\text{ }\ldots\text{ }X_{k,n_{jk}}\right]  ^{T}\\
&  =\mathbf{S}^{q-n_{jk}}X_{k}^{q}%
\end{align}
where $\mathbf{S}^{q-n_{jk}}$ is a $q\times q$ shift matrix with entries
$S_{m,n}$ that are non-zero only for $\left(  m,n\right)  =(q-n_{jk}%
+n,n),n=1,2,\ldots,n_{jk}$.

In a layered erasure IFC, we model each of the four transmit-receive links as
a $q$\textit{-bit layered erasure channel}. A $q$-bit layered erasure channel
is defined in \cite{cap_theorems:TseYatesLi_01} and summarized below.

\begin{definition}
[\cite{cap_theorems:TseYatesLi_01}]\label{Def1}A $q$-bit layered erasure
channel has input $X^{q}\in\mathbb{F}_{2}^{q}$ and output $Y^{q}=\left[
0\text{ }\ldots\text{ }0\text{ }X^{N}\right]  $ where $N$ is an integer
channel state that is independent of $X^{q}$ and satisfies $P\left[
N\geq0\right]  =1$ and $P\left[  N\geq q+1\right]  =0.$
\end{definition}

From Definition \ref{Def1}, in every use of the channel, the received signal
$Y_{j}^{q}$, $j=1,2,$ of a layered erasure IFC is given by%
\begin{equation}%
\begin{array}
[c]{cc}%
Y_{j}^{q}=\mathbf{S}^{q-N_{j1}}X_{1}^{q}\oplus\mathbf{S}^{q-N_{j2}}X_{2}^{q}, &
j=1,2,
\end{array}
\label{Y_IFC}%
\end{equation}
where $\oplus$ denotes the XOR operation, $N_{11}$ and $N_{22}$ are the random
variables representing the fading channel states over the direct links, and
$N_{21}$ and $N_{12}$ are the random variables representing the cross-link
fading states. The one-sided IFC considered in this paper is obtained by
setting $N_{12}=0$, i.e., receiver $1$ sees no interference from transmitter
$2$. One can visualize the resulting one-sided channel as an `S-IFC'.

As a first step towards developing the sum capacity of a layered erasure
one-sided IFC, we will develop the ergodic sum capacity of a two-user layered
erasure multiple-access channel (MAC), consisting of one receiver and two
transmitters. For this MAC, the received signal $Y_{j}^{q}$, $j=1$, is given
by (\ref{Y_IFC}). For simplicity, we eliminate the subscript $1$ and write
$N_{1}$ and $N_{2}$ to denote the fading states of transmitters 1 and 2 to the
receiver, respectively.

For a random variable $N\,$, we write $\Pr\left[  N=n\right]  $ to
denote the probability mass function and $\overline{F}_{N}\left(
n\right)  $ to denote the complementary cumulative distribution
function (CDF). It is straightforward to verify that
\begin{equation}
\mathbb{E}[N]=\sum_{n=1}^{q}\overline{F}_{N}(n)=\sum_{n=1}^{q}\Pr
\left[ N\geq
n\right]  . \label{EN_FN}%
\end{equation}
We also write $x^{+}=\max\left(  x,0\right)  $. All logarithms are are taken
to the base 2 and the rates are in units of bits per channel use. Throughout
the sequel we use the words transmitters and users interchangeably.
\vspace{-.1in}
\section{\label{Sec_3}Layered Fading MAC:\ sum capacity}
\vspace{-.05in}
Consider a multiple access channel with the two transmitters
transmitting $X_{1}$ and $X_{2}$ respectively, and a received
signal $Y$ given by%
\vspace{-.05in}
\begin{equation}
Y=X_{1}^{N_{1}}\oplus X_{2}^{N_{2}},
\end{equation}
where $N_{1}$ and $N_{2}$ are the channel states for the links from the two
transmitters to the receiver respectively. Both random variables $N_{1}$ and
$N_{2}$ satisfy ${\Pr}[N_{i}\ge0]=1$ and ${\Pr}[N_{i}\ge q+1]=0$.

\begin{theorem}
The capacity region for the layered erasure multiple access channel is given
by
\begin{align}
R_{1}  &  \leq\mathbb{E}[N_{1}]\\
R_{2}  &  \leq\mathbb{E}[N_{2}]\\
R_{1}+R_{2}  &  \leq\mathbb{E}[\max(N_{1},N_{2})] .
\end{align}

\end{theorem}

\begin{proof}
We will describe the achievability here since the converse is straightforward.
We prove the achievability of a corner point given by the rate pair
$(\mathbb{E}[(N_{1}-N_{2})^{+}],\mathbb{E}[N_{2}])$. The capacity region can
then be achieved by interchanging the coding schemes over the two links and by
time sharing.

The first user uses a codebook of rate $\Pr(N_{1}-N_{2}\geq n)$ to
transmit a message at level $n.$ For a given level $n$, the second
user uses a codebook of rate $\Pr(N_{2}\geq n)$ to transmit its
message. Codebooks are independent across layers at both users.
Across all channel states, i.e., on average, the receiver receives
$\mathbb{E}[1_{N_{1}-N_{2}\geq n}]=\Pr(N_{1}-N_{2}\geq n)$ bits
from level \thinspace$n$ of user 1, where we have used the fact
that the expected value of an indicator function of an event is
the probability of that event. The codebook rate of user 1 at this
level therefore allows the receiver to reliably decode the message
of user $1.$ After decoding the messages of user 1 its
contribution from the received signal can be canceled and the
remaining contribution of the second user can be decoded reliably.
Thus,
across all levels, the average transmission rates of $\mathbb{E}[(N_{1}%
-N_{2})^{+}]$ and $\mathbb{E}[N_{2}]$ at users 1 and 2, respectively, enable
reliable communications. \vspace{.1cm} \begin{figure}[ptbh]
\centering
\subfigure[Channel state 1]{
\includegraphics[
trim=1.100000in 1.100000in 2.139081in 1.100000in,
width=1.2in,
]{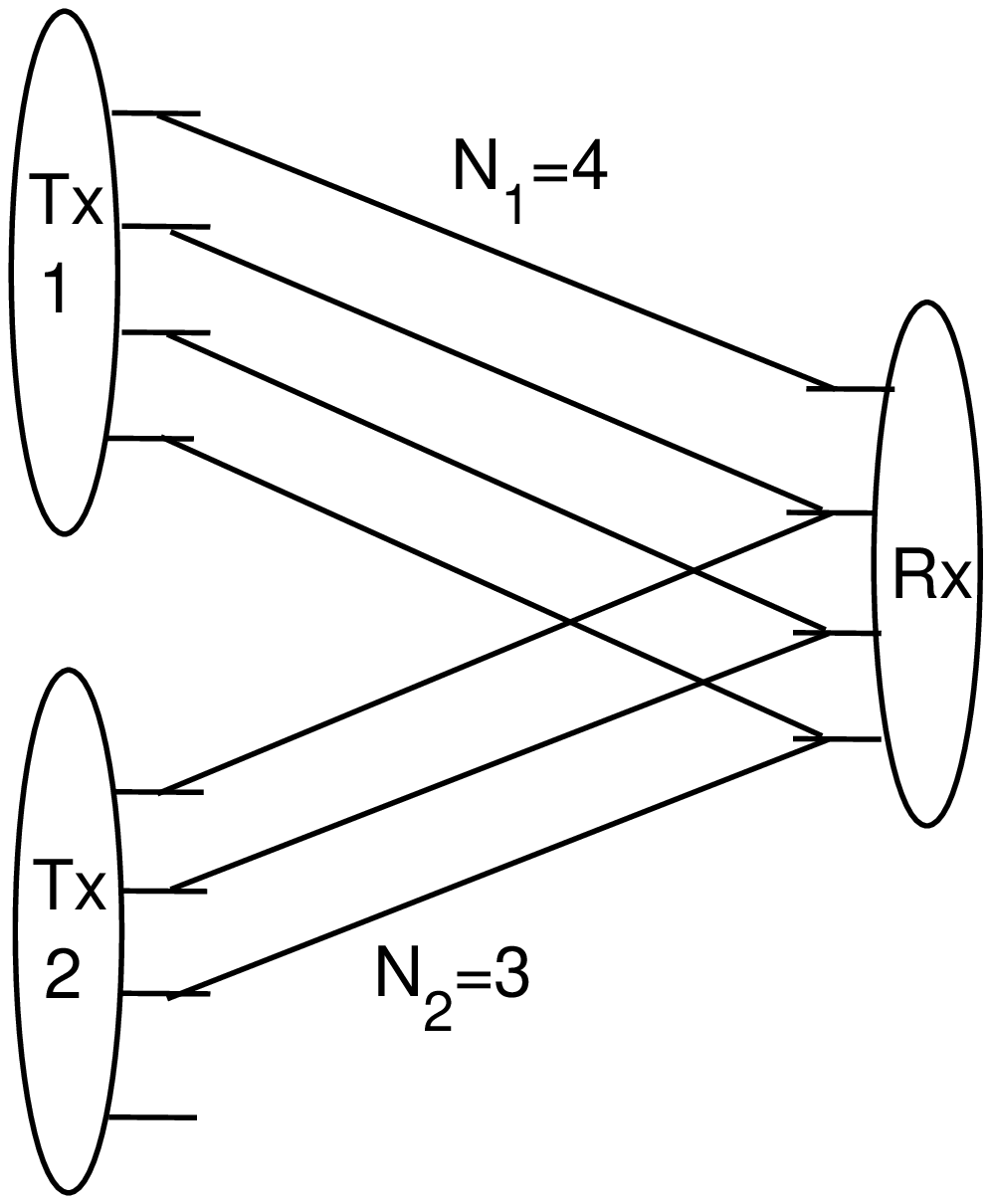}\label{Fig_Ex2}}\hspace{1cm} \subfigure[Channel state 2 ]{
\includegraphics[
trim=1.100000in 1.100000in 2.139081in 1.100000in,
width=1.2in,
]{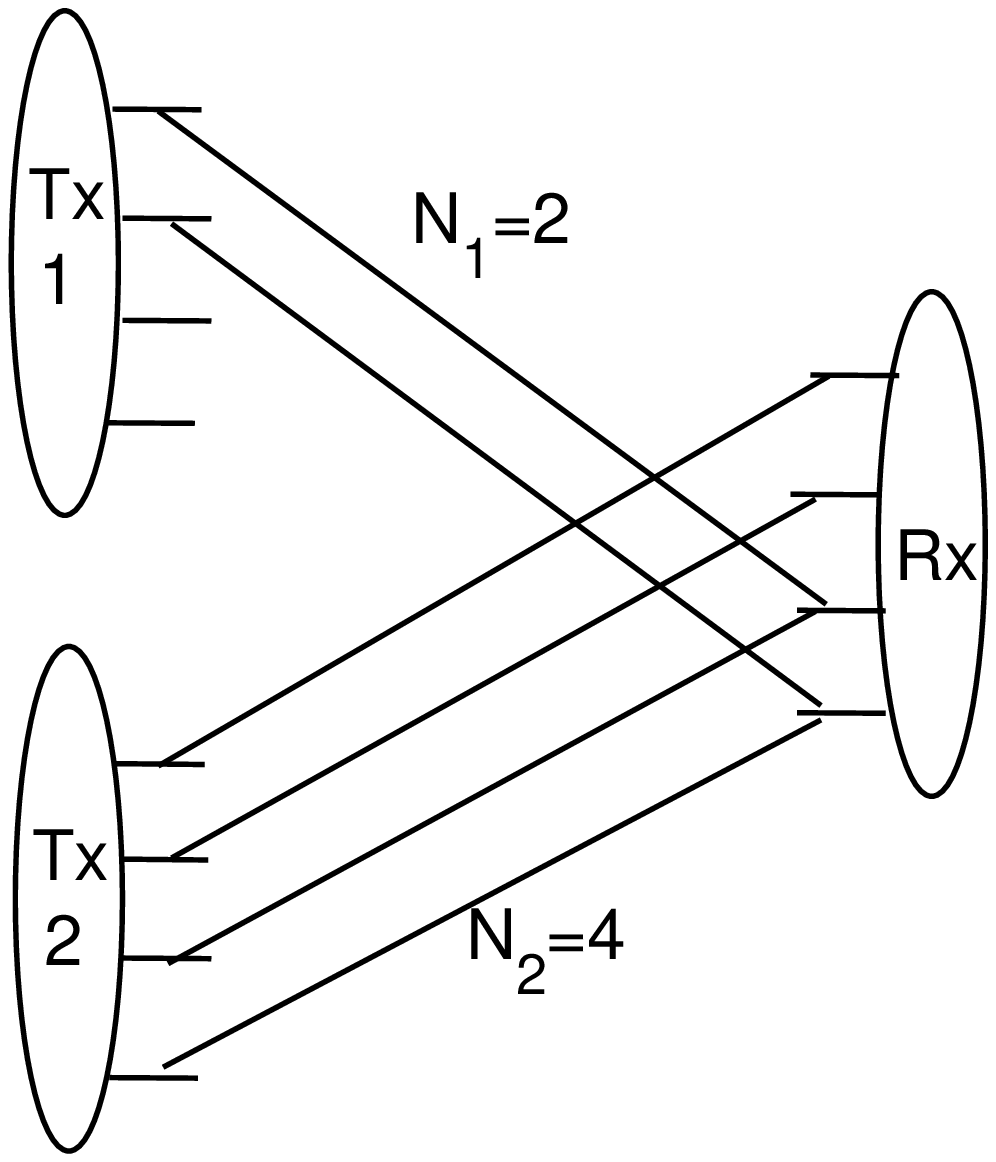}\label{Fig_Ex3}}\caption{The layered MAC in Example
\ref{exmax} in each of the two states}%
\end{figure}
\end{proof}

\begin{example}
\label{exmax} Consider a layered MAC with $q=4$ and two fading
states: the first state with $N_{1}=4$ and $N_{2}=3$ occurs with
probability $p$, and the second state with $N_{1}=2$ and $N_{2}=4$
with probability $1-p$. The above achievability scheme for rate
pair $(p,4-p)$ reduces to the following. At transmitter 1, a rate $p$
code is used on the first level while nothing is transmitted on the
remaining levels. At transmitter 2, a rate 1 code is used at the
top three levels while a rate $1-p$ code is used on the fourth
level. Note that in this case, whenever the channel is in the
first state ($N_{1}=4,N_{2}=3$) the top bit of user 1 reaches the
receiver noiselessly. Hence, the rate $p$ codeword of user 1 can be
decoded from the occurrences of state 1. Thus, the contribution of
the first transmitter can be cancelled by the receiver across both
states. Following this, the receiver uses the top 3 levels of the
second transmitter that are interference-free in both the states
and hence a rate of 1 bit/channel use can be achieved for each of
the three levels. The fourth level reaches the receiver whenever
the system is in the second state ($N_{1}=2,N_{2}=4$) which
happens with probability $1-p$, and thus the codebook of rate $1-p$
can be decoded by the receiver from the occurrences of the second
state.
\end{example}
\vspace{-.1in}

\section{\label{Sec_4}Layered Erasure One-sided IFC}
\vspace{-.05in}

\subsection{Outer Bounds}
\vspace{-.05in}

Outer bounds on the capacity region of a class of deterministic IFCs, of which
the binary expansion deterministic IFC is a sub-class, are developed in
\cite{cap_theorems:CostaElGamal_IC2}. For a time-varying (ergodic)
layered erasure IFC with perfect CSI at the receivers, we follow the
same steps as in \cite[Theorem 1]{cap_theorems:CostaElGamal_IC2} while
including the CSI as a part of the received signal at each receiver. The
following theorem summarizes the outer bounds on the capacity region of
layered erasure one-sided IFCs.

\begin{theorem}
\label{Th_OB}An outer bound of the capacity region of an ergodic layered
erasure one-sided IFC is given by the set of all rate tuples $\left(
R_{1},R_{2}\right)  $ that satisfy
\begin{subequations}
\label{COB}%
\begin{align}
R_{1}  &  \leq\mathbb{E}[N_{11}]\\
R_{2}  &  \leq\mathbb{E}[N_{22}]\\
R_{1}+R_{2}  &  \leq\mathbb{E}[\max(N_{11},N_{22},N_{21},N_{11}+N_{22}%
-N_{21})].
\end{align}

\end{subequations}
\end{theorem}
\subsection{Optimality of Outer Bounds}

We now prove the tightness of the sum capacity outer bounds for specific
sub-classes of ergodic layered erasure IFCs. For the very strong sub-class,
the achievable scheme also achieves the capacity region. For the remaining
sub-classes, we achieve a corner point of the capacity region.

\subsubsection{Very Strong IFC}

\begin{theorem}
\label{vs} For a class of very strong layered erasure IFCs for which
$N_{21}\geq N_{11}+N_{22}$ holds with probability 1, the sum capacity is
$\mathbb{E}[N_{11}+N_{22}]$ and the capacity region is given by $R_{1}%
\leq\mathbb{E}[N_{11}]$ and $R_{2}\leq\mathbb{E}[N_{22}]$.
\end{theorem}

\begin{proof}
Consider the following achievable scheme: at level $n$, the first
user uses a codebook of rate
$\Pr(\min(N_{11},(N_{21}-N_{22})^{+})\geq n)=\Pr(N_{11}\geq n)$,
i.e., at each level, the first user transmits at the erasure rate
supported by that level at the receiver. On the other hand, at
level $n$, the second user uses a codebook of rate $\Pr(N_{22}\geq
n)$ to transmit its message. At both users, encoding is
independent across layers. The message of user $1$ can be reliably
decoded at receiver $1$ and the average rate achieved is
\begin{equation}
R_{1}=\sum\limits_{n=1}^{q}\Pr(N_{11}\geq n)=\sum_{n=1}^{q}\overline{F}_{N_{11}%
}(n)=\mathbb{E}[N_{11}].
\end{equation}
The second receiver acts like a multi-access receiver and at each
level, it first decodes the message of user $1$. Thus, across all
channel states, it can, on average, reliably decode the message
from level $n$ at a rate $\Pr(\min(N_{11},(N_{21}-N_{22})^{+})\geq
n)$. After decoding all levels of user $1$, receiver $2$
eliminates the contribution of user $1$ from its received signal
thereby decoding the messages from user $2$ interference-free at
an average rate of $\mathbb{E}[N_{22}]$.
\end{proof}

\subsubsection{Strong but not Very Strong IFC}

\begin{theorem}
The sum capacity of a class of very strong layered erasure IFCs for which
$N_{11}\leq N_{21}\leq N_{11}+N_{22}$ with probability 1 is $\mathbb{E}%
[\max(N_{21},N_{22})]$.
\end{theorem}

\begin{proof}
The proof is very similar to that of Theorem \ref{vs}, and is hence omitted.
\end{proof}

\subsubsection{Strong IFC}

For the two sub-classes considered thus far, it sufficed to use independent
coding across the layers. However, for the sub-class with a mix of SnVS and VS
states, joint coding across the layers is required as shown in the following theorem.

\begin{theorem}
If $N_{21}\geq N_{11}$ with probability 1, the sum capacity is given by
$\min(\mathbb{E}[N_{11}+N_{22}],\mathbb{E}[\max(N_{21},N_{22})])$.
\end{theorem}

\begin{proof}
Let $\mathbb{E}[(N_{21}-N_{22})^{+}]\leq\mathbb{E}[N_{11}]$. In this case, the
first user forms a codebook of rate $\mathbb{E}[(N_{21}-N_{22})^{+}]/q$. The
transmitter at level $n$ sends data from this codebook while the second user
at level $n$ uses a codebook of rate $\Pr(N_{22}\geq n)$ to transmit the data.

The decoding scheme proceeds as follows. The first receiver
receives across all channel states, i.e., on average,
$\sum_{n=1}^{q}\mathbb{E}[1_{N_{11}\geq
n}]=\sum_{n=1}^{q}\Pr(N_{11}\geq n)=\mathbb{E}[N_{11}]$ bits from
all the levels of user 1 and is thus able to decode data at the
lower rate of $\mathbb{E}[(N_{21}-N_{22})^{+}]/q$. Similarly, the
second receiver receives across all channel states, i.e., on
average, $\sum_{n=1}^{q}\mathbb{E}[1_{N_{21}-N_{22}\geq
n}]=\sum_{n=1}^{q}\Pr(N_{21}-N_{22}\geq
n)=\mathbb{E}[(N_{21}-N_{22})^{+}]$ bits reliably from all the
levels of user 1 and is thus able to decode. After decoding user
$1$, receiver $2$ eliminates the contribution of user $1$ from its
received signal thereby decoding the
messages from user $2$ interference-free at an average rate of $\mathbb{E}%
[N_{22}]$.

One can proceed similarly for $\mathbb{E}[(N_{21}-N_{22})^{+}]\geq\mathbb{E}[N_{11}]$. In this
case, the first user forms a codebook of rate $\mathbb{E}[N_{11}]/q$ and the
same strategy achieves the sum capacity.
\end{proof}

\subsubsection{Ergodic Very Strong IFC}

More generally, one can also consider the sub-class of IFCs with a mix of all
types of sub-channels, i.e., a mix of weak, SnVS, and VS. In the
following theorem we develop the sum capacity for subset of such a sub-class
in which on average the conditions for very strong are satisfied.

\begin{theorem}
If $\mathbb{E}[\max(N_{21},N_{22})]\geq\mathbb{E}[N_{{11}}+N_{22}]$, then the sum
capacity is $\mathbb{E}[N_{11}+N_{22}]$.
\end{theorem}

\begin{proof}
The first user forms a codebook of rate $\mathbb{E}[N_{11}]/q$.
The transmitter at level $n$ sends data from this codebook while
the second user at level $n$ uses a codebook of rate
$\Pr(N_{22}\geq n)$ to transmit the data. The first receiver
receives across all channel states, i.e., on average, it receives
$\mathbb{E}[N_{11}]$ bits from all the levels of user 1 and is
thus able to decode. Similarly, the second receiver receives
across all channel states, i.e., on average, it receives
$\mathbb{E}[(N_{21}-N_{22})^{+}]\ge \mathbb{E}[N_{11}]$ bits from
all the levels of user 1 and is thus able to decode. After
decoding user $1$, receiver $2$ eliminates the contribution of
user $1$ from its received signal thereby decoding the messages
from user $2$ interference-free at an average rate of
$\mathbb{E}[N_{22}]$.
\end{proof}

\subsubsection{Weak IFC}

\begin{theorem}
If $N_{21}\leq N_{11}$ with probability 1, then the sum capacity is $\mathbb{E}%
(\max(N_{11},N_{11}+N_{22}-N_{21}))$.
\end{theorem}

\begin{proof}
Consider the following achievable scheme: at level $n$, the first
user uses a codebook of rate $\Pr(N_{11}\geq n)$, i.e., at each
level, the first user transmits at the erasure rate supported by
that level at its receiver. On the other hand, at level $n$, the
second user uses a codebook of rate $\Pr (N_{22}-N_{21}\geq n)$ to
transmit its message. The second receiver receives across all
channel states, i.e., on average, it receives $\mathbb{E}[1_{N_{22}-N_{21}\geq
n}]=\Pr(N_{22}-N_{21}\geq n)$ bits reliably from all level $n$ of
user 2 and is thus able to decode.
\end{proof}

\subsubsection{Mixed IFC}

\begin{theorem}
\label{Th_Mixed}For a layered erasure one-sided IFC, the following sum rate
can be achieved:%
\begin{equation}
\mathbb{E}[N_{22}]+\sum_{n=1}^{q}(\Pr(N_{11}\geq n)-\Pr(N_{21}\geq n,N_{21}%
-N_{22}<n))^{+}. \label{Mixed_SR}%
\end{equation}

\end{theorem}

\begin{proof}
The transmission scheme is as follows: user $1$ transmits on a subset
$\mathcal{I}_{1}$ of levels on which it is more likely to be received at its
intended receiver than it is to interfere with user $2$, i.e.,
\begin{multline*}
\mathcal{I}_{1}=\{n\in\lbrack1,q]:\Pr(N_{11}\geq n)\geq\\
\Pr(N_{21}\geq n,N_{21}-N_{22}<n)\}.
\end{multline*}
Furthermore, user 1 transmits at level $n$ (independent coding across levels)
using a codebook of rate $\Pr(N_{11}\geq n)$ for $n\in\mathcal{I}_{1}$ and
does not transmit on the remaining levels such that%
\begin{equation}
R_{1}=\sum_{n\in\mathcal{I}_{1}}\Pr(N_{11}\geq n).
\end{equation}
Since user $1$ is transmitting at the erasure rate for any level
$n\in\mathcal{I}_{1}$, receiver 1 can decode the data of the first
transmitter with asymptotically negligible error probability. The
second user transmits a message encoded across all layers. This in
turn allows receiver $2$ to decode the message of the second user
jointly across those layers that do not experience interference
from the first user on average. Consider a level
$n\in\mathcal{I}_{1}$ at the first transmitter. This level
interferes with the data of the second user at the second receiver
when $N_{21}\geq n$ and $N_{21}-N_{22}<n$. Thus, all the levels of
the first user interfere on an average with
$\mathbb{E}[\sum_{n\in\mathcal{I}_{1}}1_{N_{21}\geq
n,N_{21}-N_{22}<n}]=%
%TCIMACRO{\tsum _{n\in\mathcal{I}_{1}}}%
%BeginExpansion
{\textstyle\sum_{n\in\mathcal{I}_{1}}}
%EndExpansion
\Pr(N_{21}\geq n,N_{21}-N_{22}<n)$ bits. Hence, for reliable reception,
transmitter $2$ needs to transmit at an average rate\
\begin{equation}
R_{2}=\mathbb{E}[N_{22}-\sum_{n\in\mathcal{I}_{1}}\Pr(N_{21}\geq
n,N_{21}-N_{22}<n)]
\end{equation}
bits/channel use across all levels. The sum-rate is then given by
(\ref{Mixed_SR}).
\end{proof}

\begin{lemma}
\label{condifc} For every $n\in\lbrack1,q],$ let
\begin{subequations}
\label{Lemma_Cond}%
\begin{align}
&
\begin{array}
[c]{cc}
& A_{1}\left(  n\right)  =1_{(N_{21}<n\leq N_{11})\cup\left(  \min(N_{11}%
,N_{21}-N_{22})\geq n\right)}
\end{array}
\\
&
\begin{array}
[c]{cc} & A_{2}\left(  n\right)  =1_{(N_{11}<n,N_{21}\geq
n,N_{21}-N_{22}<n)}
\end{array}
\\
&
\begin{array}
[c]{cc}%
s.t. & \left(  \mathbb{E}\left[  A_{1}\left(  n\right) \right]
-\mathbb{E}\left[ A_{2}\left( n\right) \right]  \right)
^{+}=\mathbb{E}\left[ A_{1}\left( n\right)  \right]  .
\end{array}
\end{align}
Given (\ref{Lemma_Cond}), the sum-rate in (\ref{Mixed_SR}) simplifies to
\end{subequations}
\begin{multline}
\mathbb{E}[\min(N_{11}+N_{22}+(N_{11}-N_{21})^{+},\\
\max(N_{11},N_{21},N_{22},N_{11}+N_{22}-N_{21}))].
\end{multline}

\end{lemma}

\begin{remark}
Choosing $N_{11}$ as deterministic is a sufficient condition for
(\ref{Lemma_Cond}).
\end{remark}

\begin{proof}
Consider the $n^{th}$ term in the summation over $\mathcal{I}_{1}$ in
(\ref{Mixed_SR}) in Theorem \ref{Th_Mixed}. Using the fact that for any two
sets $\mathcal{A}$ and $\mathcal{B}$, $\Pr(\mathcal{A})-\Pr(\mathcal{B}%
)=\Pr(\mathcal{A}\backslash\mathcal{B})-\Pr(\mathcal{B}\backslash\mathcal{A})$,
we
have%
\begin{equation}%
\begin{array}
[c]{c}%
(\Pr(N_{11}\geq n)-\Pr(N_{21}\geq n,N_{21}-N_{22}<n))^{+}=\\
(\Pr(N_{21}<n\leq N_{11}\cup\min(N_{11},N_{21}-N_{22})\geq n)\\
-\Pr(N_{11}<n\leq N_{21},N_{21}-N_{22}<n))^{+}.
\end{array}
\label{Mix_Simp}%
\end{equation}
Substituting (\ref{Lemma_Cond}) into (\ref{Mix_Simp}), every term within the
summation in (\ref{Mix_Simp}) simplifies as
\begin{align}
\Pr(N_{21}  &  <n\leq N_{11}\cup\min(N_{11},N_{21}-N_{22})\geq n)\nonumber\\
=  &  \Pr(N_{21}<n\leq N_{11})+\Pr(\min(N_{11},N_{21}-N_{22})\geq
n).
\end{align}
Summing over all $n\in\lbrack1,q]$ and adding $\mathbb{E}[N_{22}]$, the
sum-rate in (\ref{Mixed_SR}) then simplifies as%
\begin{align}
&  \mathbb{E}[N_{22}]+\mathbb{E}[(N_{11}-N_{21})^{+}]\\
&  +\mathbb{E}[\min(N_{11},(N_{21}-N_{22})^{+})]\\
&  =\mathbb{E}[\max(N_{11}-N_{21},0)]\nonumber\\
&  \text{ \ \ \ }+\mathbb{E}[\min(N_{11}+N_{22},\max(N_{21},N_{22}))]\\
&  =\mathbb{E}[\min(N_{11}+N_{22}+(N_{11}-N_{21})^{+},\nonumber\\
&  \text{ \ \ \ \ }\max(N_{11},N_{21},N_{22},N_{11}+N_{22}-N_{21}))].
\end{align}

\end{proof}

\begin{theorem}
\label{Th_Mix_SC}The sum capacity of a class of mixed layered erasure IFCs for
which the condition (\ref{Lemma_Cond}) of Lemma \ref{condifc} is satisfied and
$N_{21}\leq N_{11}+N_{22}$ with probability 1 is given by
\begin{equation}
\mathbb{E}[\max(N_{11},N_{21},N_{22},N_{11}+N_{22}-N_{21})].
\end{equation}

\end{theorem}

\begin{example}
\label{exevs}(Ergodic Very Strong) Consider a layered IFC with $q=4$ and two
fading states: the first state with $N_{11}=2,$ $N_{21}=1$ and $N_{22}=4$ occurs
with probability $1/2$, and the second state with $N_{11}=1,$ $N_{21}=4$ and
$N_{22}=1$ with probability $1/2$. The first state is weak while the second is
very strong, but overall the net mixture is ergodic very strong. Thus, the sum
capacity of $4$ bits/channel use can be attained.
\end{example}

We now present two examples for the mixed IFC. For the first, the sum capacity
is given by Theorem \ref{Th_Mix_SC}; for the second, we present a new sum
capacity achieving strategy.

\begin{example}
\label{exmix}(Mixed) Consider a layered IFC with $q=4$ and two fading states:
the first state with $N_{11}=2,$ $N_{21}=1$ and $N_{22}=2$ occurs with
probability $1/2$, and the second state with $N_{11}=3,$ $N_{21}=4$ and $N_{22}=1$
with probability $1/2$. The first state is weak while the second is strong,
but overall the net mixture satisfies all the conditions in Theorem
\ref{Th_Mix_SC}. (Note that although $N_{11}$ is not deterministic, the
condition in Lemma \ref{condifc} is satisfied.) Thus, the ergodic sum capacity
of $7/2$ bits/channel use can be attained. \begin{figure}[ptbh]
\centering
\subfigure[State 1]{
\includegraphics[
trim=1.100000in 1.100000in 2.139081in 1.100000in,
width=1in,
]{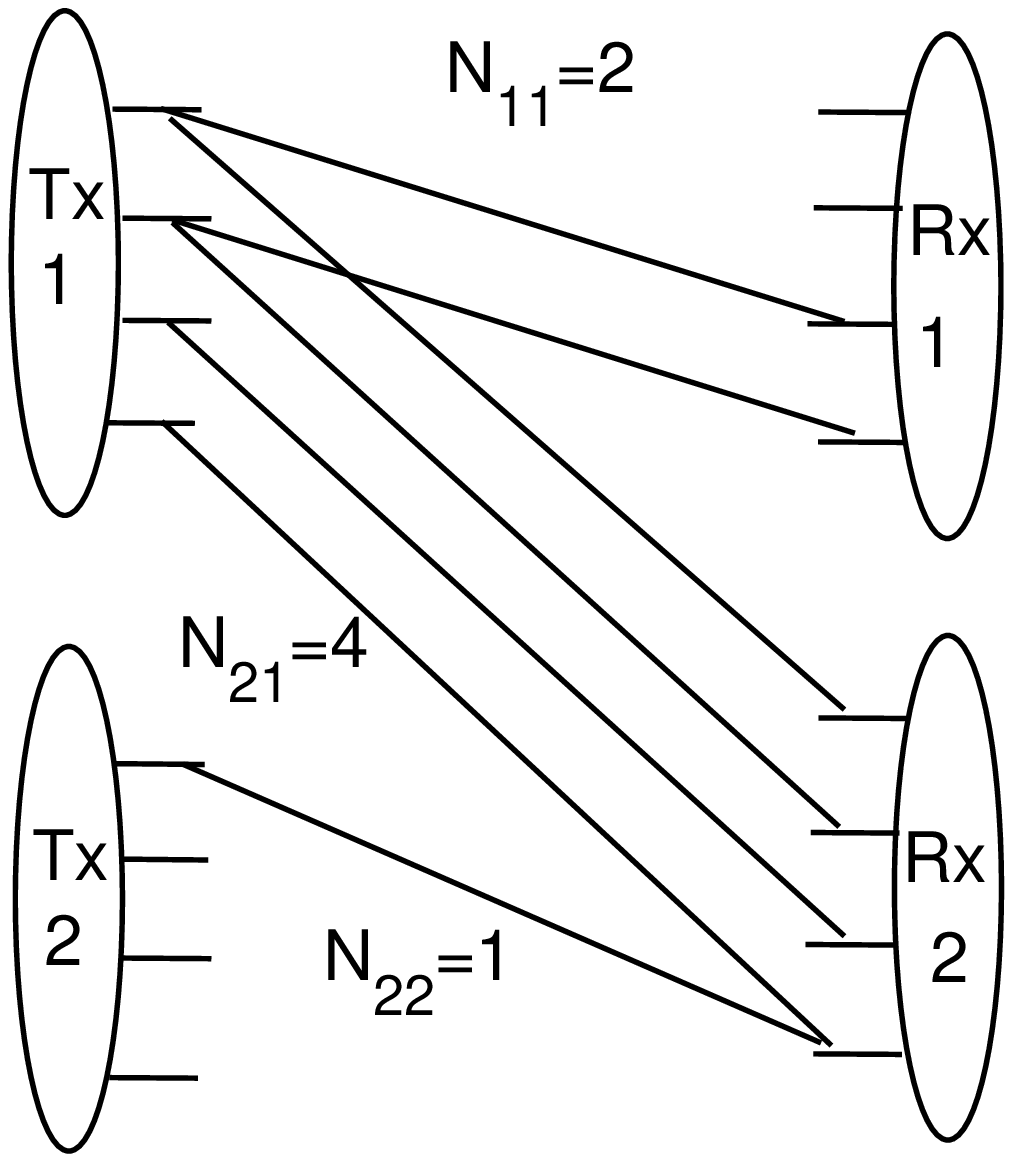}\label{Fig_Ex2}}\hspace{1cm} \subfigure[State 2 ]{
\includegraphics[
trim=1.100000in 1.100000in 2.139081in 1.100000in,
width=1in,
]{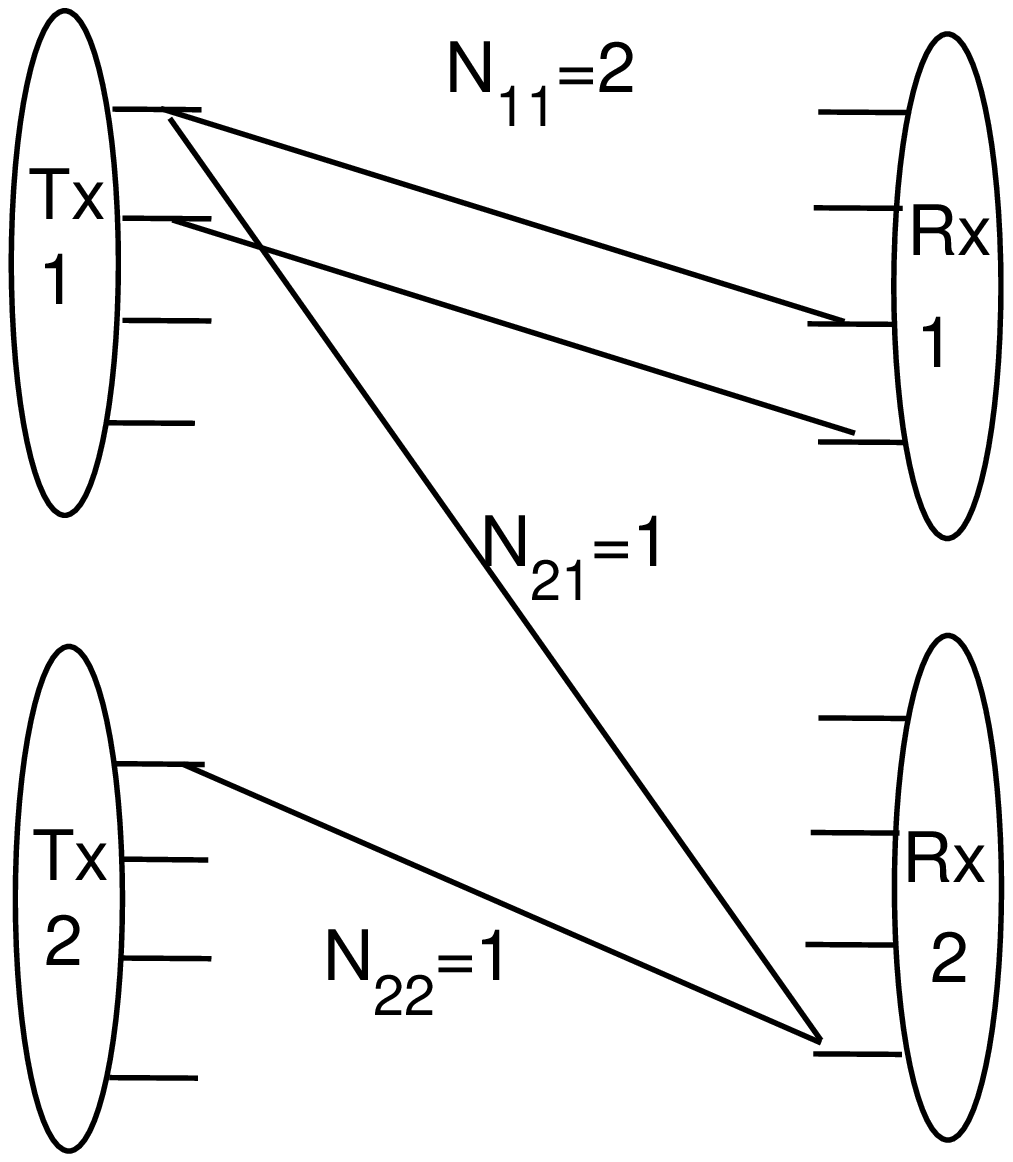}\label{Fig_Ex3}}\caption{The layered IFC in Example
\ref{exifc} in each of the two states}%
\end{figure}
\end{example}

\begin{example}
\label{exifc} (Mixed) Consider a layered IFC with $q=4$ and two fading states:
the first state with $N_{11}=2,$ $N_{21}=4$ and $N_{22}=1$ occurs with
probability $1/2$, and the second state with $N_{11}=2,$ $N_{21}=1$ and $N_{22}=1$
with probability $1/2$. The first state is very strong while the second is
weak though the IFC is not ergodic very strong. The states satisfy the
condition in Lemma \ref{condifc}, and thus, the sum rate of 5/2 bits/channel
use can be achieved. However, applying Theorem \ref{Th_OB} the outer bound on
sum capacity is 3 bits/channel use. We here present an alternate achievable
strategy that achieves this outer bound. At its second level, the first
transmitter sends a message at a rate of 1 bit/channel use which its intended
receiver can always decode but the second receiver cannot. Suppose receiver 2
does not decode this second level in either channel state. Thus, with respect
to receiver 2, the equivalent channel has two fading states: the first state
$N_{11}=1,$ $N_{21}=3$ and $N_{22}=1$ with probability $1/2$, and the second state
$N_{11}=1,$ $N_{21}=1$ and $N_{22}=1$ with probability $1/2$. This is an ergodic
strong IFC and hence a sum capacity of 2 bits/channel use can be achieved.
Combining that with the rate sent to receiver 1 from the second level of
transmitter 1, we achieve a sum capacity of 3. Note that our strategy uses a
public and a private message from the first transmitter at the first and
second levels, respectively. Thus, while the second level from the first
transmitter is received at the second receiver half of the time, the message
on this level is considered private from the second user. This is in contrast
with the deterministic interference channel where the message reaching the other
receiver is always public.

%We further note in this example that the separate coding even with
%perfect transmit CSI will not be optimal since separate coding
%will yield sum rate of 5/2 even with perfect transmit CSI while
%sum rate of 3 is achievable with no transmit CSI by coding across
%states.\vspace{-0.1in}
\end{example}
%\vspace{-.1in}

\section{\label{Sec_5}Concluding Remarks}
%\vspace{-.05in}

%\vspace{-0.1in}
 We have developed inner and outer bounds on the
sum capacity of a class of layered erasure ergodic fading IFCs. We
have shown that the outer bounds are tight for the following
sub-classes: i) weak, ii) strong, iii) \textit{ergodic very
strong} (mix of strong and weak), and (iv) a sub-class of mixed
interference (mix of SnVS and weak), where each sub-class is
uniquely defined by the fading statistics. Our work demonstrates
that for layered erasure IFCs with sub-channels that are not
uniquely of one kind, i.e., that are not all strong but not very
strong or very strong or weak, joint encoding is required across
layers. Of immediate interest is extending these results to the
ergodic fading Gaussian IFCs without transmitter CSI. Furthermore,
we are also exploring extending the results of Theorem
\ref{Th_Mix_SC} to both general layered IFCs as well as ergodic
fading Gaussian IFCs. %\vspace{-0.1in}

%\vspace{-.15in}

\end{document}